\documentclass[aps,pra,showpacs,amssymb,superscriptaddress,nofootinbib,notitlepage,11pt]{revtex4-1}

\usepackage{theorem,setspace,bm,enumerate,graphicx,epic,eepic,epsfig,amsmath,latexsym,amssymb,verbatim,revsymb,color,stmaryrd,fancyhdr,hyperref}

\newtheorem{definition}{Definition}

\newtheorem{lemma}[definition]{Lemma}

\newtheorem{theorem}[definition]{Theorem}
\newtheorem{corollary}[definition]{Corollary}
\newtheorem{conjecture}[definition]{Conjecture}

\def\squareforqed{\hbox{\rlap{$\sqcap$}$\sqcup$}}
\def\qed{\ifmmode\squareforqed\else{\unskip\nobreak\hfil
\penalty50\hskip1em\null\nobreak\hfil\squareforqed
\parfillskip=0pt\finalhyphendemerits=0\endgraf}\fi}
\def\endenv{\ifmmode\;\else{\unskip\nobreak\hfil
\penalty50\hskip1em\null\nobreak\hfil\;
\parfillskip=0pt\finalhyphendemerits=0\endgraf}\fi}

\newlength{\blank}
\newenvironment{proof}[1][{\hspace{-\blank}}]{{\noindent\textbf{Proof   }}}{\hfill\qed\vskip 0.5\baselineskip}

\mathchardef\ordinarycolon\mathcode`\:
\mathcode`\:=\string"8000
\def\vcentcolon{\mathrel{\mathop\ordinarycolon}}
\begingroup \catcode`\:=\active
  \lowercase{\endgroup
  \let :\vcentcolon
  }

\newcommand{\nc}{\newcommand}
\nc{\rnc}{\renewcommand}
\nc{\binomial}[2]{\genfrac{(}{)}{0pt}{}{#1}{#2}}
\nc{\lbar}[1]{\overline{#1}}
\nc{\bra}[1]{\langle#1|}
\nc{\ket}[1]{|#1\rangle}
\nc{\ketbra}[2]{|#1\rangle\!\langle#2|}
\nc{\braket}[2]{\langle#1|#2\rangle}
\nc{\rank}{\operatorname{rank}\,}
\nc{\tr}{\operatorname{Tr}}

\nc{\cA}{{\cal A}}
\nc{\cB}{{\cal B}}
\nc{\cC}{{\cal C}}
\nc{\cD}{{\cal D}}
\nc{\cE}{{\cal E}}
\nc{\cF}{{\cal F}}
\nc{\cG}{{\cal G}}
\nc{\cH}{{\cal H}}
\nc{\cI}{{\cal I}}
\nc{\cJ}{{\cal J}}
\nc{\cK}{{\cal K}}
\nc{\cL}{{\cal L}}
\nc{ \cM}{{\cal M}}
\nc{\cN}{{\cal N}}
\nc{\cO}{{\cal O}}
\nc{\cP}{{\cal P}}
\nc{\cR}{{\cal R}}
\nc{\cS}{{\cal S}}
\nc{\cT}{{\cal T}}
\nc{\cU}{{\cal U}}
\nc{\cX}{{\cal X}}
\nc{\cZ}{{\cal Z}}

\nc{\1}{{\openone}}

\def\d{\delta}
\def\e{\epsilon}

\def\m{\mu}

\def\r{\rho}
\def\s{\sigma}

\nc{\rR}{{{\mathbb R}}}
\nc{\CC}{{{\mathbb C}}}
\nc{\FF}{{{\mathbb F}}}
\nc{\NN}{{{\mathbb N}}}
\nc{\ZZ}{{{\mathbb Z}}}
\nc{\PP}{{{\mathbb P}}}
\nc{\QQ}{{{\mathbb Q}}}
\nc{\UU}{{{\mathbb U}}}
\nc{\EE}{{{\mathbb E}}}
\nc{\id}{{\operatorname{id}}}

\nc{\be}{\begin{equation}}
\nc{\ee}{\end{equation}}
\nc{\bea}{\begin{eqnarray}}
\nc{\eea}{\end{eqnarray}}

\nc{\LO}{\text{LO}}
\nc{\LOCC}{\text{LOCC}}
\nc{\cLOCC}{{\overline{\text{LOCC}}}}
\nc{\rEP}{\text{SEP}}
\nc{\PPT}{\text{PPT}}
\nc{\rep}{\text{sep}}
\nc{\twist}{\text{twist}}
\nc{\te}{\otimes}
\nc{\pro}[1]{\ket{#1}\!\bra{#1}}
\nc{\mc}{\mathcal}

\begin{document}
 \singlespacing
\title{Every entangled state provides an advantage in classical communication}
\author{Stefan B{\"a}uml}
\email{s.m.g.bauml-1@tudelft.nl}
\affiliation{NTT Basic Research Laboratories, NTT Corporation, 3-1 Morinosato-Wakamiya, Atsugi, Kanagawa 243-0198, Japan}
\affiliation{NTT Research Center for Theoretical Quantum Physics, NTT Corporation, 3-1 Morinosato-Wakamiya, Atsugi, Kanagawa 243-0198, Japan}
\affiliation{QuTech, Delft University of Technology, Lorentzweg 1, 2628 CJ Delft, Netherlands}
\author{Andreas Winter}
\email{andreas.winter@uab.cat}
\affiliation{Departament de F\'{i}sica, Grup d'Informaci\'{o} Qu\`{a}ntica, Universitat Aut\`{o}noma de Barcelona, ES-08193 Bellaterra (Barcelona), Spain}
\affiliation{ICREA - Instituci\'{o} Catalana de Recerca i Estudis Avan\c{c}ats, ES-08010 Barcelona, Spain}
\author{Dong Yang}
\email{dyang@cjlu.edu.cn}
\affiliation{Department of Informatics, University of Bergen, 5020 Bergen, Norway}
\affiliation{Laboratory for Quantum Information, China Jiliang University, Hangzhou, Zhejiang 310018, China}

\begin{abstract}
We investigate the use of noisy entanglement as a resource in classical communication via a quantum channel. In particular, we are interested in the question whether for \emph{any} entangled state, including bound entangled states, there exists a quantum channel the classical capacity of which can be increased by providing the state as an additional resource. We partially answer this question by showing, for any entangled state, the existence of a quantum memory channel the feedback-assisted classical capacity with product encodings of which can be increased by using the state as a resource. Using a different (memoryless) channel construction, we also provide a sufficient entropic condition for an advantage in classical communication (without feedback and for general encodings) and thus provide an example of a state that is not distillable by means of one-way local operations and classical communication (LOCC), but can provide an advantage in the classical capacity of a number of quantum channels. As separable states cannot provide an advantage in classical communication, our condition also provides an entropic entanglement witness.
\end{abstract}

\date{\today}
\maketitle

\section{Introduction}
Since the early days of quantum information theory it is known that maximally entangled states can increase the rate of classical communication via a noiseless quantum channel. By means of a fundamental protocol known as \emph{superdense coding} \cite{bennett1992communication} it is possible to send two classical bits in a single use of a noiseless qubit channel, assisted by a maximally entangled state: The sender, Alice, can encode her bits by transforming the maximally entangled state into any of the four Bell basis states by means of local unitary operations. The channel is then used to send Alice's part of the Bell state to Bob, who can then distinguish between the four basis states by means of a Bell state measurement. Despite being a very elegant protocol, this version of entanglement assisted classical communication is set in an idealised scenario, where both the quantum channel and the assisting entanglement are noiseless. A more realistic and generally applicable scenario would use general noisy channels and mixed state entanglement. 

One direction of research, which has been pursued, is to go beyond noiseless channels and study the classical communication capacity of general quantum channels, while still requiring pure entanglement assistance. In a setting allowing for an arbitrary amount of maximal entanglement assistance, a capacity theorem has been derived \cite{bennett1999entanglement}. It states that the entanglement assisted classical capacity of a quantum channel is equal to the maximal \emph{mutual information} obtainable between Alice and Bob by sending part of a pure state over the given quantum channel. 

Research has also been conducted in another direction: Instead of pure, maximally entangled states the authors of \cite{bose2000mixed,hiroshima2001optimal,bowen2001classical,horodecki2001classical,winter2002scalable,horodecki2012quantum} have investigated the advantage given by mixed states, but in the setting of classical communication via a noiseless quantum channel. In \cite{bowen2001classical}, it has been shown that the advantage a mixed entangled state $\r_{AB}$ can provide in dense coding is determined by its \emph{coherent information}, defined by $I(A\rangle B)_\r=S(B)_\r-S(AB)_\r$. If $I(A\rangle B)_\r\le0$ the state $\r$ cannot provide an advantage. The setting assumed is restricted to unitary encodings that are independent for each channel use and one copy of $\r_{AB}$ per channel use. In \cite{horodecki2001classical,horodecki2012quantum}, the setting has been generalised to arbitrary encodings. It has been shown, however, that even in this generalised scenario the advantage which $\r_{AB}$ can provide is determined by the maximal coherent information obtainable from $\r_{AB}$ by means of local operations on Alice's side.
The case where an arbitrary number of copies of $\r_{AB}$ is available per channel use, as in \cite{bennett1999entanglement}, has also been considered \cite{horodecki2012quantum}. In this case the coherent information needs to be regularised.

As the coherent information plays an important role in determining the advantage entanglement can provide in dense coding, we believe it will be instructive to review another operational meaning of this quantity. It has been shown \cite{DevetakWinter-hash} that the coherent information $I(A\rangle B)_\r$ provides a lower bound on the asymptotic rate at which $\r_{AB}$ can be distilled to maximal entanglement. Hence any state that is not one-way distillable will not be of use in classical communication via a noiseless channel \cite{horodecki2012quantum}. An example of a one-way undistillable state is the two-qubit Werner state 
\be\label{eq:Werner}
\frac{q}{3}P_\text{sym}+(1-q)P_\text{anti},
\ee
with $q\ge 1/4$. Here $P_\text{sym}$ and $P_\text{anti}$ denote the projectors onto the symmetric and antisymmetric subspace, respectively. Going further, there exists a large class of states which cannot be distilled, even if two-way communication is available. Such states are known as \emph{bound entangled}  \cite{H97,Horodecki:1998kf}. The resource character of bound entangled states for various information theoretic tasks is still an active field of research. For example, it has been shown that any bound entangled state can be activated to improve the teleportation fidelity as long as a seed of distillable entanglement is available \cite{horodecki1999bound,M06}. Another task where any bound entangled state can provide an advantage is channel discrimination \cite{piani2009all}. Further, bound entanglement has been shown to be potentially useful for quantum key distribution \cite{HHHO052,HHHO05} as well as a source for Bell non-locality \cite{vertesi2014disproving}, while being only of limited use in quantum repeaters \cite{bauml2015limitations}.

In the present paper we combine the two research directions mentioned above. We investigate the scenario in which both the quantum channel and the assisting entanglement are noisy. This doubly noisy scenario has also been considered in \cite{zhuang2016additive}. Apart from being the most realistic scenario experimentally, the doubly noisy scenario also poses an interesting open question concerning the resource character of states that are not one-way distillable, in particular bound entangled states. Namely, we are interested whether for any entangled state there exists some quantum channel such that the state can provide an advantage in classical communication via the channel. For states with positive coherent information, this question has already been answered above; such states can yield an advantage in classical communication via the noiseless channel. However, for entangled states with vanishing or negative coherent information, in particular states which are only two-way distillable and bound entangled states, it is not known whether such channels exist. Hence, what we are looking for is a kind of activation effect, where a combination of a given noisy entangled state and a particular noisy channel provide a higher rate of classical communication than the channel alone, whereas the state combined with a noiseless channel cannot provide an advantage. Let us also note that shared randomness as well separable states, which can be simulated by shared randomness, cannot provide an advantage in classical communication \cite{prevedel2011entanglement}.

The results we present in this work partially answer the question posed above, i.e. whether for any entangled state there exists a channel the classical capacity of which can be increased by the state. In a first approach, we consider channels with finite memory, and communication schemes that allow for feedback (noiseless classical backward communication) after each channel use. We show that for each entangled state, there exists a corresponding memory channel, the feedback-assisted classical capacity with product encodings of which can be increased by using the state as an additional resource.

The second approach we present in this paper is set in the usual framework of many independent channel uses (i.e. a memoryless channel), and without feedback. Given a given state $\r_{\tilde A\tilde B}$ and a channel $\mc{M}:\tilde A\to B$, we provide an entropic condition which is sufficient for the existence of another channel $\mc{N}:\tilde AD\to B$, which can be constructed from $\mc{M}$, such that the Holevo capacity of $\mc{N}$ can be increased if we use $\r_{\tilde A\tilde B}$ as an additional resource. The condition is of the form
\be\label{eq:EntCond}
S(B|\tilde B)_\omega<S_{\text{min}}(\mc{M}),
\ee
where $\omega_{B\tilde B}=\mc{M}\te \id\left(\r_{\tilde A\tilde B}\right)$ and $S_{\text{min}}$ denotes the minimum output entropy. The channel $\mc{N}$ is constructed from $\mc{M}$ in a way which was introduced in \cite{shor2004equivalence} in the context of equivalences of various additivity questions. The same construction has also been used in \cite{zhu2017superadditivity,zhu2018superadditivity}. As only entangled states can provide an advantage in the Holevo capacity, the condition (\ref{eq:EntCond}) with a suitable choice of $\mc{M}$ can also serve as an entropic entanglement witness for the state $\r_{\tilde A\tilde B}$. If, in addition to fulfilling condition (\ref{eq:EntCond}), $\mc{M}$ is entanglement breaking, $\mc{N}$ will be entanglement breaking as well, hence the Holevo quantity will equal the capacity. Using condition (\ref{eq:EntCond}) we can show that the two-qubit Werner state (\ref{eq:Werner}) can provide an advantage in the capacity of a number of channels for $q\le0.345$, i.e. for values where it is not one-way distillable, hence useless for dense 
coding \cite{horodecki2012quantum}. 

The organisation of this paper is as follows: In section \ref{sec:feedback} we present our first approach using backward communication. Section \ref{sec:Shor} contains our second approach, containing the entropic condition. In the final section we summarise our work and discuss open questions and future directions. Part of this work, in particular section \ref{sec:Shor}, are part of SB's Ph.D. thesis \cite{bauml2015applications}.

\section{Channels and resources for communication}
In this section, we introduce the communication scenario in which we will identify an advantage for every entangled state. Namely, we will review the plain, unassisted classical capacity (subsection \ref{II-A}) and the classical capacity assisted by a pre-shared entangled state (subsection \ref{II-B}), both for memoryless stationary ("i.i.d.") channels. As we cannot, as of yet, prove general separations in this setting, we extend the allowed coding schemes to (classical) feedback between the successive channel uses (subsections \ref{II-C} and \ref{II-D}), and finally consider channels with memory (subsection \ref{II-E}).


\subsection{Unassisted classical capacity}\label{II-A}
We begin by shortly reviewing some concepts of unassisted classical communication via quantum channels. For an in-depth treatment of these topics see \cite{MarkBook}. Assume that Alice intends to send a classical message $m\in M$ to Bob via $n$ uses of a quantum channel $\mc{N}:A\to B$. To to so, she prepares a state $\sigma_{A^n}^m$ and inputs each subsystem into the channel. Bob receives $\mc{N}^{\te n}\left(\sigma_{A^n}^m\right)$, on which he performs a POVM $\left\{E_m\right\}$, element $E_m$ corresponding to message $m$. The rate of communication is defined as 
\be
R=\frac{1}{n}\log|M|,
\ee 
and the error probability is given by
\be
P_e=1-\frac{1}{|M|}\sum_m \tr E_m \mc{N}^{\te n}\left(\sigma_{A^n}^m\right).
\ee
A rate of communication $R$ is said to be achievable if, for all $\e>0$ and large enough $n$, there exists a coding scheme with $P_e\le\e$. Now we can define the \textit{classical capacity of a quantum channel} as
\be\label{eq:classCap}
\cC(\mc{N})=\sup\left\{R:R\text{ is an achievable rate}\right\}.
\ee 
By the Holevo-Schuhmacher-Westmoreland (HSW) theorem \cite{1996quant.ph.11023H,PhysRevA.56.131}, it holds
\be
\cC(\mc{N})=\chi^{\infty}(\mc{N})=\lim_{n\to\infty}\frac{1}{n}\chi\left(\mc{N}^{\te n}\right),
\ee 
where $\chi(\mc{N})$ is the Holveo capacity of the channel $\mc{N}$, defined by
\be\label{def:Holevo1}
\chi(\mc{N})=\max_{\{p_x,\ket{\psi_x}\}} I(X:B)_\sigma,
\ee 
where
\be
\sigma_{XB}=\sum_x p_x\pro{x}_X\te \mc{N}\left(\pro{\psi_x}\right)_B
\ee
is known as the cq-state. $X$ denotes the random variable corresponding to the codewords $\ket{\psi_x}$ which Alice enters into the channel. An equivalent definition is given by
\be\label{def:Holevo2}
\chi(\mc{N})=\max_{\{p_x,\ket{\psi_x}\}}\left[S\left(\mc{N}\left(\sum_xp_x\pro{\psi_x}\right)\right)-\sum_xp_xS\left(\mc{N}\left(\pro{\psi_x}\right)\right)\right].
\ee
A rate of $\chi(\mc{N})$ can be achieved by Alice choosing a product encoding 
\be
\sigma_{A^n}^m=\pro{\psi_{x_1(m)}}_{A_1}\te...\te\pro{\psi_{x_n(m)}}_{A_n},
\ee
for some codebook $\left\{\textbf{x}(m)\right\}_{m\in M}$, agreed by Alice and Bob.

\subsection{Noisy entanglement assisted capacity}\label{II-B}
Assume now that, in addition to the $n$ channel uses, Alice and Bob share $n$ copies of an entangled state $\r_{\tilde A \tilde B}$. Alice can encode her message $m\in M$ by applying an encoding map and sending the state via $n$ uses of the channel $\mc{N}:A\to B$. Bob can then apply a collective measurement to the $\tilde B^nB^n$ subsystem. This scenario is demonstrated in Figure \ref{fig:EntAs}. The maximal achievable rate we call \textit{$\rho$-assisted classical capacity}. 
\be\label{eq:AssClassCap}
\cC_{\rho} (\mc{N})=\sup\left\{R:R\text{ is an achievable rate using }\rho_{\tilde A\tilde B}\right\},
\ee
where $R=\frac{\log\left| M\right|}{n}$ is the rate of communication. For separable $\rho$, eq. (\ref{eq:AssClassCap}) reduces to the unassisted classical capacity \cite{prevedel2011entanglement}.

Let us also define a quantity analogous to the Holevo capacity, eq. (\ref{def:Holevo1}), but including assistance by a state $\rho_{\tilde A\tilde B}$. Instead of Alice's possible input ensembles, the maximisation is performed over ensembles of encoding maps $\Lambda_x^{\tilde A\to A}$, which Alice can apply to her share of $\rho_{\tilde A\tilde B}$ before using the channel. Namely, we define
\be
\chi_{\rho} (\mc{N})=\max_{\s^{cq}} I(X:B\tilde B)_{\s^{cq}},
\ee
with cq state 
\be
\s^{cq}_{XCB}=\sum_xp_x\pro{x}_X\te\mc{N}^{A\to B}\circ\Lambda_x^{\tilde A\to A}(\rho_{\tilde A\tilde B}).
\ee
By the achievability part of the HSW theorem \cite{1996quant.ph.11023H,PhysRevA.56.131}, it holds $\cC_{\rho} (\mc{N})\ge \chi_{\rho} (\mc{N})$. 

\begin{figure}
	\centering
	\includegraphics[width=0.4\textwidth]{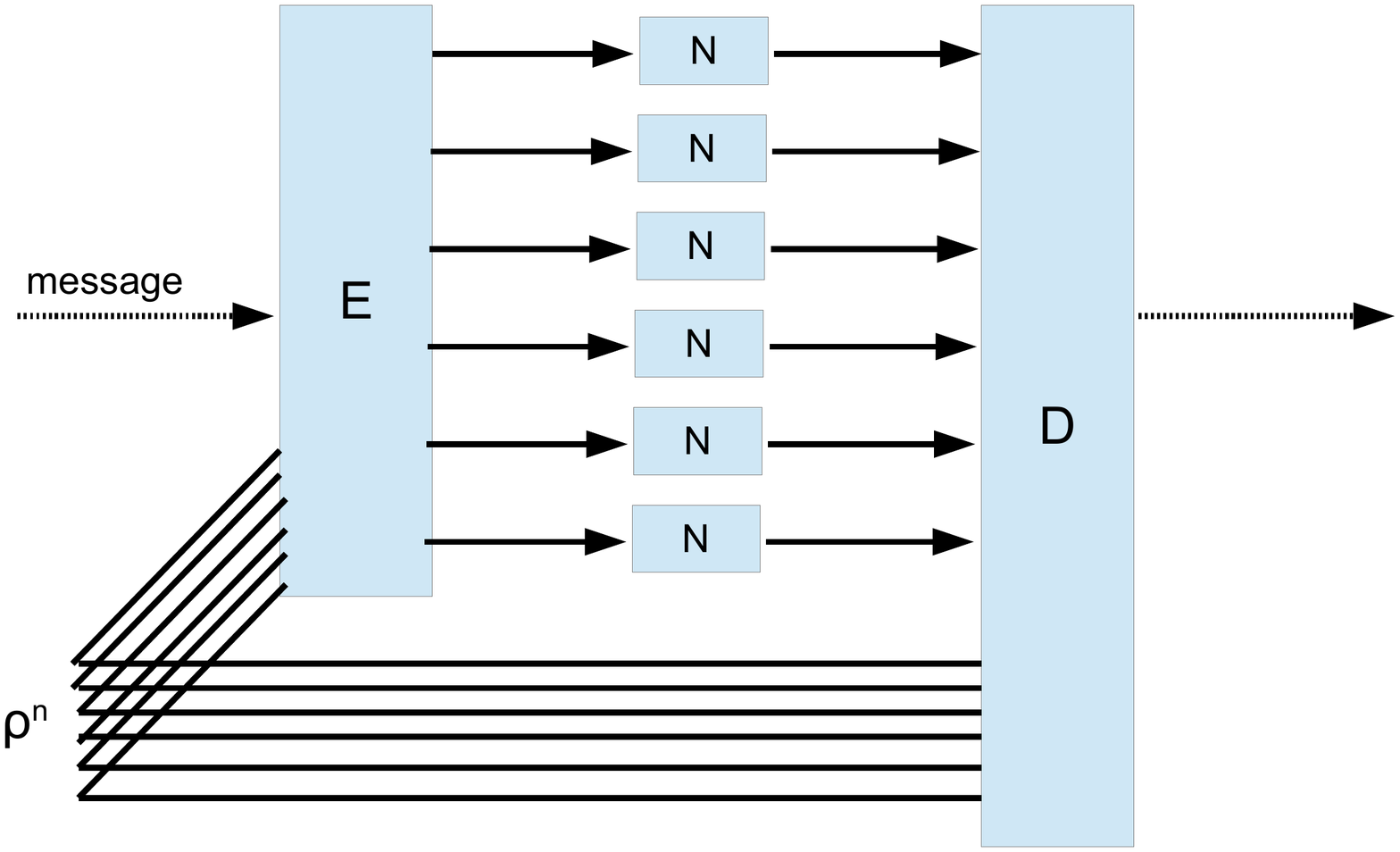}
		\caption{Our scenario: Alice and Bob have $n$ copies of $\rho$ and can use the channel $n$ times.}\label{fig:EntAs}
\end{figure}

\subsection{Feedback-assisted capacity}\label{II-C}
In this scenario we assume that Bob can send classical information back to Alice for free, allowing for the following strategy: Upon receiving the first output of the channel $\mc{N}^{A\to B}$, Bob can perform a quantum instrument, i.e. a set of completely positive maps $\{\cE^{(1)}_{j_1}\}$, with $\cE^{(1)}_{j_1}:B_1\to B'_1$, such that $\sum_{j_1}\cE^{(1)}_{j_1}$ is trace preserving, and send the outcome $j_1$ back to Alice as a classical message, while keeping the quantum output state of the instrument. After receiving the second channel output Bob can apply quantum instrument $\{\cE^{(2)}_{j_2}\}$ jointly on the output state of the channel and the output of his previous instrument $\{\cE^{(1)}_{j_1}\}$, i.e. $\cE^{(2)}_{j_2}:B'_1B_2\to B'_2$. The classical outcome $j_2$ is then sent back to Alice, while the output state is kept by Bob for the next round, and so on. After receiving the last output, Bob can perform his usual a decoding operation to guess Alice's message.

Alice can use the feedback messages $j_1, j_2,...,j_{n-1}$ which she receives from Bob to choose her input states into the channel $\mc{N}^{A\to B}$ accordingly. In this work we restrict ourselves to product encodings, i.e. the input states, w.r.t. a priori distribution $\{p_x\}$, are products of the form
\be\label{eq:productEncoding}
\omega_{A_1}^{(1,x)}\te\omega_{A_2}^{(2,j_1x)}\te\omega_{A_3}^{(3,j_2x)}\te\cdots\te\omega_{A_n}^{(n,j_{n-1}x)},
\ee 
where $A_i$ is the $i$-th input system. See also figure \ref{fig:FB}. The restriction to product encodings is a realistic assumption if Alice does not posses a quantum memory with coherence time long enough to wait for the classical feedback. Note that we could, in principle, allow Alice to use the feedback she receives in a given round to modify the encoding not only for the next, but for all following rounds, thus creating classical correlations between the input systems. However, since we allow for arbitrary classical feedback, such a strategy can be simulated by Bob storing all outcomes and including all previous outcomes $j_1,\cdots,j_{i-1}$ into each feedback message $j_i$. This can be easily achieved by choosing Bobs instruments to have additional classical in- and output registers. Let us now define the the \emph{feedback-assisted classical product capacity} as
\be\label{eq:feedMemCap}
\cC_{\te}^{\leftarrow}(\mc{N})=\sup\left\{R:R\text{ achievable rate with product encodings and free feedback}\right\}.
\ee 
Because of correlations that can be created by the feedback operations we cannot, in general, use the normal Holevo capacity (\ref{def:Holevo1}). However, we can recursively define a Holevo-like quantity that depends on the number of channel uses. Namely, for $n$ channel uses and $n-1$ feedbacks, we define 
\be
\chi^{(n)}(\mc{N})=\max I(X:B'_{n-1}B_n)_{\s^{(n)}},
\ee
where the maximisation is over a priori distributions $\{p_x\}_x$, instruments $\left\{\cE^{(1)}_{j_1}\right\}_{j_1},\cdots,\left\{\cE^{(n-1)}_{j_{n-1}}\right\}_{j_{n-1}}$ and sets of product input states $\left\{\omega_{A_1}^{(1,x)}\te\omega_{A_2}^{(2,j_1x)}\te\cdots\te\omega_{A_n}^{(n,j_{n-1}x)}\right\}_{x,j_1,\cdots,j_{n-1}}$ and the cq-state $\s^{(n)}$ is given by
\be
\s^{(n)}_{XB'_{n-1}B_n}=\sum_xp_x\pro{x}_X\te\s^{(n,x)}_{B'_{n-1}B_n},
\ee
where $\s^{(n,x)}$ is defined recursively by (choosing $\text{dim}(B'_0)=1$)
\be
\s^{(1,x)}_{B_1}=\mc{N}(\omega_{A_1}^{(1,x)})
\ee
and for $n>1$,
\be
\s^{(n,x)}_{B'_{n-1}B_n}=\sum_{j_{n-1}}\cE^{(n-1)}_{j_{n-1}}\left(\s^{(n-1,x)}_{B'_{n-2}B_{n-1}}\right)\te\mc{N}^{A_n\to B_n}\left(\omega_{A_n}^{(n,j_{n-1}x)}\right).
\ee
In the case of \emph{memoryless} quantum channels it has been shown \cite{bowen2004feedback,bowen2005feedback} that 
\be
\chi^{(n)}(\mc{N})\leq n\chi(\mc{N}),
\ee
showing that feedback cannot improve the classical capacity with product encodings. In the case of channels with memory (see subsection \ref{II-E}), however, it is not known whether feedback can provide an advantage.

\begin{figure}
	\includegraphics[width=0.45\linewidth, trim={2cm 6cm 2cm 2cm},clip=true]{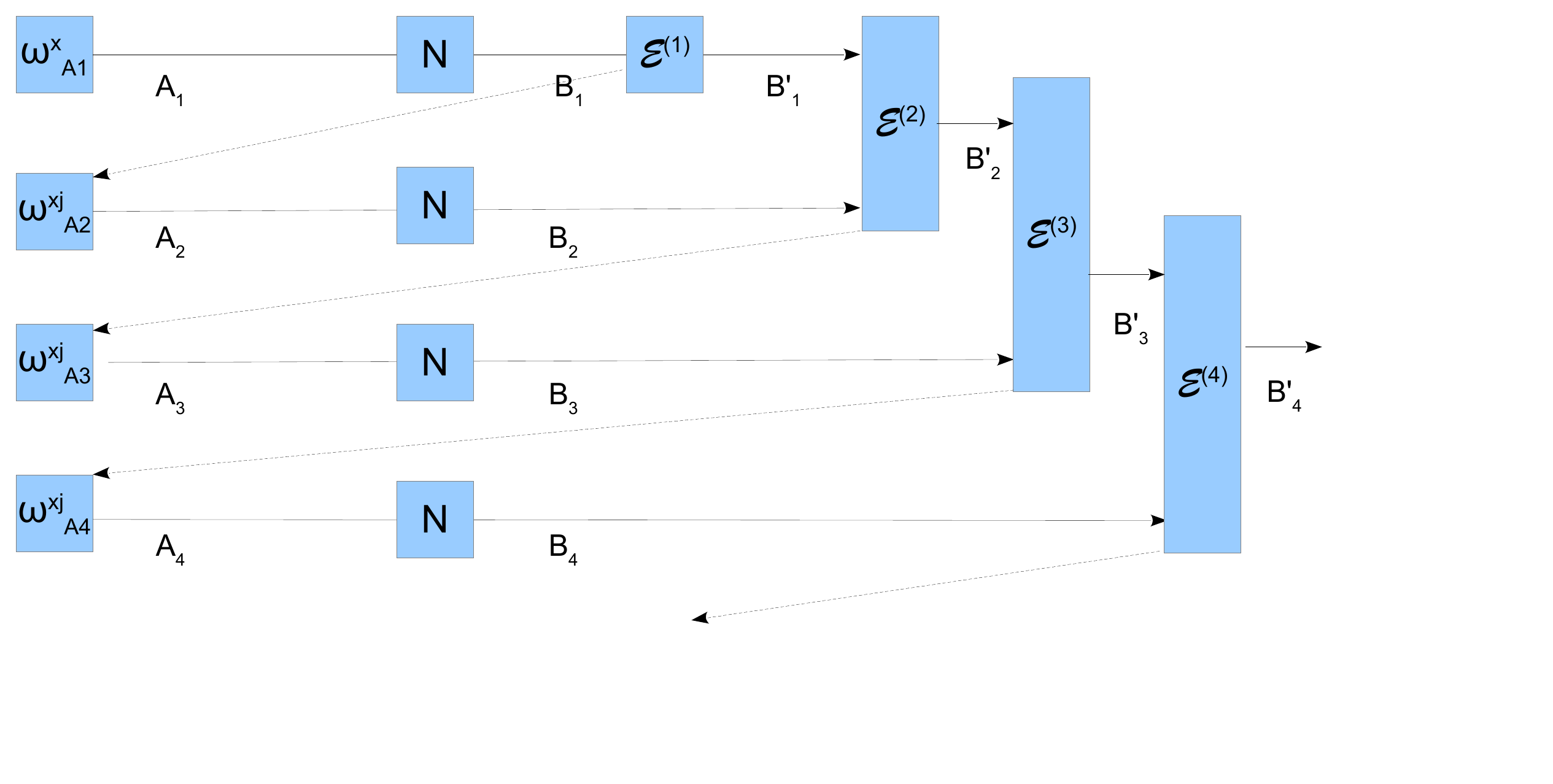}
	\includegraphics[width=0.45\linewidth, trim={2cm 6cm 2cm 2cm},clip=true]{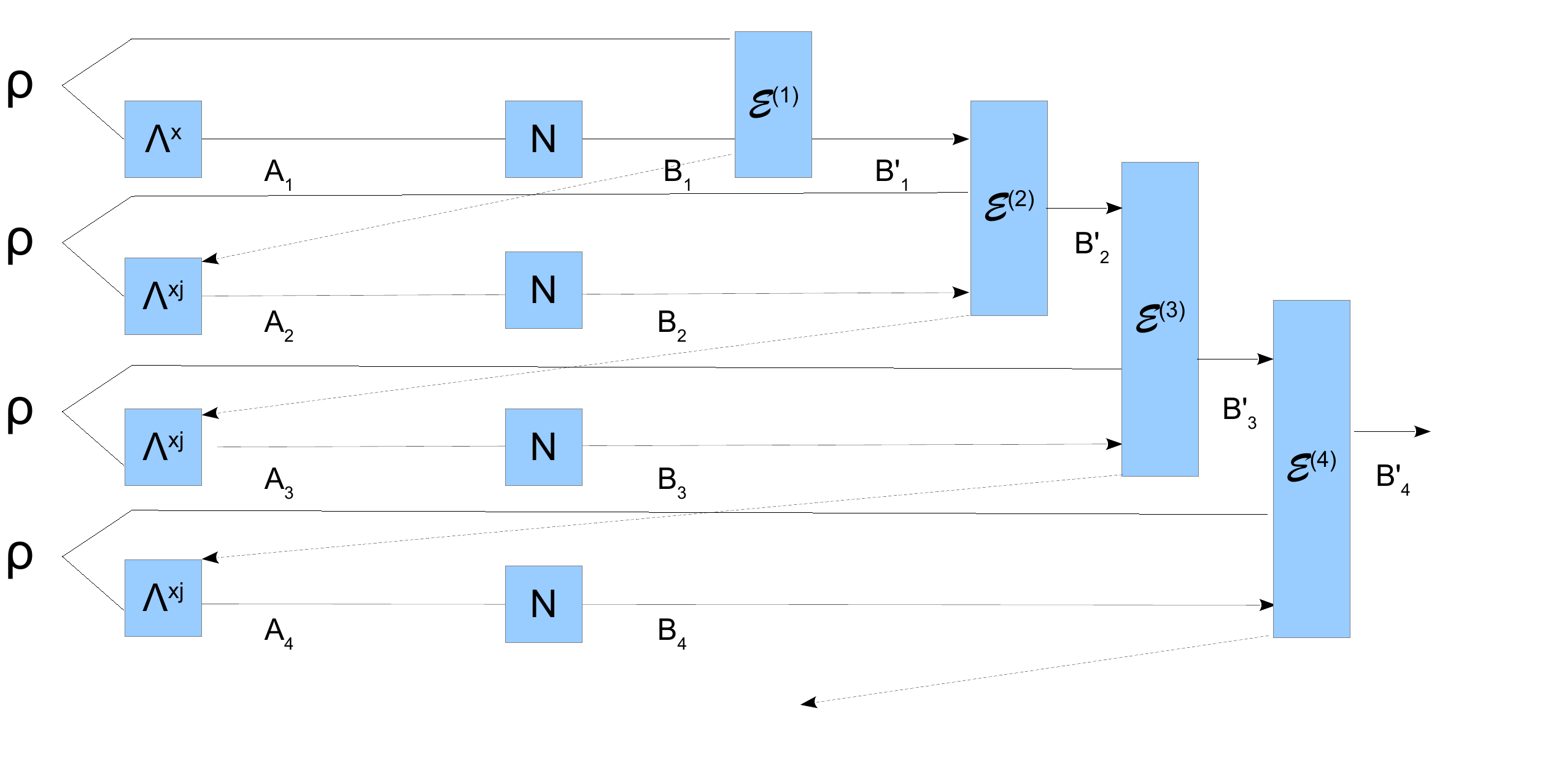}

		\caption{Scheme of the first four channel uses in our feedback-assisted protocol. The l.h.s. shows the case without $\r$-assistance, the r.h.s. shows the case with $\r$-assistance.  After $n$ channel uses Bob can perform his usual decoding operation, which is not shown here.}\label{fig:FB}
		\end{figure}

\subsection{Noisy entanglement and feedback-assisted capacity}\label{II-D}
Combining the settings of the previous two sections, we here consider the scenario where, in addition to free feedback, Alice and Bob are provided with $n$ copies of an entangled state $\r_{\tilde{A}\tilde{B}}$ for $n$ channel uses. We define a feedback-assisted protocol with product encoding using $\r$ analogously to the the protocol without $\r$ defined in the previous section, except that Alice, instead of preparing product states to be entered into the channel, applies encoding maps 
\be
\Lambda_{\tilde{A}_1\to A_1}^{(1,x)}\te\Lambda_{\tilde{A}_2\to A_2}^{(2,j_1x)}\te\Lambda_{\tilde{A}_3\to A_3}^{(3,j_2x)}\te\cdots\te\Lambda_{\tilde{A}_n\to A_n}^{(n,j_{n-1}x)},
\ee 
where, as in (\ref{eq:productEncoding}), $j_i$ are the feedback messages, to her shares of the $n$ copies of $\r_{\tilde{A}\tilde{B}}$; and Bob uses his shares of the first $n-1$ copies of $\r_{\tilde{A}\tilde{B}}$ as additional inputs to his instruments 
\begin{align}
&\left\{\cE^{(1)}_{j_1}:B_1\tilde{B}_1\to B'_1\right\}_{j_1},\\
&\left\{\cE^{(2)}_{j_2}:B_2B'_1\tilde{B}_2\to B'_2\right\}_{j_2},\\
&\cdots\\
&\left\{\cE^{(n-1)}_{j_{n-1}}:B_{n-1}B'_{n-2}\tilde{B}_{n-1}\to B'_{n-1}\right\}_{j_{n-1}},
\end{align}
and the last one as an additional input to his final decoding operation. See also figure \ref{fig:FB}. As maximum achievable rate of classical communication using this scheme, we define the \emph{$\r$-assisted classical feedback product capacity} of a channel $\mc{N}$ as
\be\label{eq:feedMemCapRho}
\cC_{\te,\r}^{\leftarrow}(\mc{N})=\sup\left\{R:R\text{ is an achievable rate with product encodings, free feedback and using }\r\right\}.
\ee 

\subsection{Channels with memory}\label{II-E}

We will now move beyond the scenario of independent channel uses, considering channels with finite classical memory \cite{bowen2004quantum,kretschmann2005quantum}. Such channels are of the form
\be\label{eq:memChan}
\mc{N}:A\te M\to M\te B,
\ee
where $M$ denotes the memory register and $A$ and $B$ the input and output systems, respectively. The initial state of the memory can be set by either Alice or a third party, Eve. The final state of the memory, after many uses of the channel, can be given to either Bob or Eve. In this work, we use the setting where the initial value is provided by Eve and the final value is given to Eve. In addition to memory, we also allow for free feedback (i.e. unlimited classical backward communication from Bob to Alice) after each use of the channel. Whereas the advantage feedback can or cannot provide in classical communication over memoryless channels has been investigated, little is known about the use of feedback in the presence of memory.\pagebreak

\section{Communication via memory channels\protect\\ assisted by free backward communication}\label{sec:feedback}

In this section we present our main result, namely we show that for any entangled state $\r_{\tilde{A}\tilde{B}}$ there exists a memory channel, the $\r$-assisted classical feedback product capacity (\ref{eq:feedMemCapRho}) of which is strictly larger than its unassisted classical feedback product capacity (\ref{eq:feedMemCap}). We can show this result by borrowing from the theory of channel discrimination. Namely, it has been shown \cite{piani2009all} that for any entangled state $\r_{\tilde{A}\tilde{B}}$, there exist two channels $m_0$ and $m_1$ that can be better distinguished with assistance of$\r_{\tilde{A}\tilde{B}}$ than with separable inputs. Our main idea is to employ this phenomenon in a scheme of communication where Alice and Bob are presented a mixture of two pairs of channels $m^{A'\to B'}_0\te n^{A''\to B''}_0$ and $m^{A'\to B'}_1\te n^{A''\to B''}_1$, such that, in a first step, they can send a probe state from $A'$ to $B'$ in order to determine which pair they have been presented with and, in a second step, are allowed to modify the encoding of $A''$ accordingly. In such a scheme, an advantage in distinguishing $m_0$ and $m_1$ can provide an advantage for classical communication from $A''$ to $B''$. By introducing a memory that stores the value $i=0,1$ aa well as an even odd counter, i.e. a memory that changes its value after every channel use, we can construct a memory channel that incorporates this communication scheme. The memory channel is of the form
\be
T^d:A\te I \te K\to I \te K\te B,
\ee
with
\be\label{eq:discChan}
T^d(\s_A\te\pro{i}_{I}\te\pro{k}_{K})=\begin{cases} 
      \pro{i}_{I}\te\pro{k\oplus1}_{K}\te m_i(\s)_B&\text{ if }k=0\\
      \frac{1}{2}\1_{I}\te\pro{k\oplus1}_{K}\te n_i^d(\s)_B&\text{ if }k=1.
   \end{cases}
\ee
Here $I$ and $K$ are two single-bit memory registers, respectively. Register $K$ contains the information whether the channel has been used an even or odd number of times. The bit value $i$ of $I$ is initially set to a random value, which neither Alice nor Bob know. During each channel use, the value $i$ determines which of a pair of two channels $m_{0,1}$ (in odd rounds) or $n^d_{0,1}$ (in even rounds) is used. Thus the two channels become correlated. After each even round, $K$ is randomised again. 

As the channels $n^d_{0,1}$, which are used in even rounds, we choose qc-channels that are defined by
\be
n^d_i(\s)=\sum_{k=1}^d\ketbra{k}{v_k^{(i)}}\s\ketbra{v_k^{(i)}}{k},
\ee
where $d\in\NN$ and $\cB_0=\{\ket{v_k^{(0)}}\}$ as well as $\cB_1=\{\ket{v_k^{(1)}}\}$ are mutually unbiased bases (MUBs), i.e. $\left|\braket{v_k^{(0)}}{v_\ell^{(1)}}\right|^2=\frac{1}{d}$. If a message is encoded in basis $\cB_i$, channel $n^d_i$ can achieve a rate of $\log d$. If, on the other hand, the encoding is in basis $\cB_{i\oplus1}$, the message will be completely depolarised by $n^d_i$. If presented with a random mixture of $n^d_{0,1}$, Alice will a priori not know which of the MUBs to encode her message in. 

If feedback is allowed, however, the structure of (\ref{eq:discChan}) allows Alice and Bob to perform the following protocol: In every odd round, Alice sends some state that helps Bob to distinguish between channels $m_{0,1}$ by means of a two outcome POVM $\{Q, \1-Q\}$. The result, $j$, is then sent back to Alice, who, in the following even round, applies a unitary operation modifying the encoding. Without loss of generality $U_0=\1$ and $U_1=\sum_{k=1}^d\ketbra{v_k^{(1)}}{v_k^{(0)}}$ is the unitary transforming $\cB_0$ to $\cB_1$. See also Figure \ref{timeflow}. This protocol is then performed many times, after which Bob decodes, as usual. 

\begin{figure}

	\centering
	\includegraphics[width=0.45\textwidth]{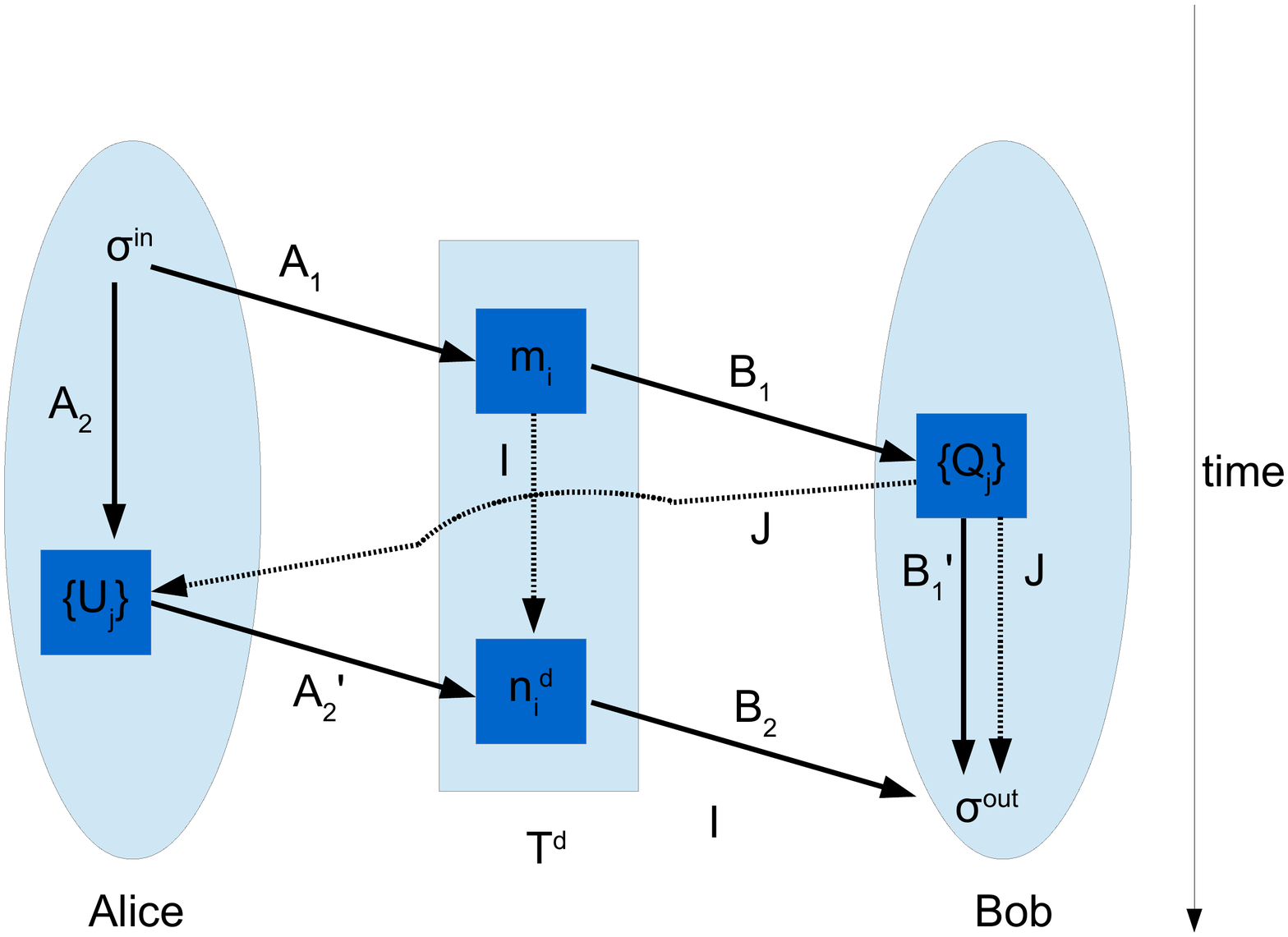}
	\includegraphics[width=0.45\textwidth]{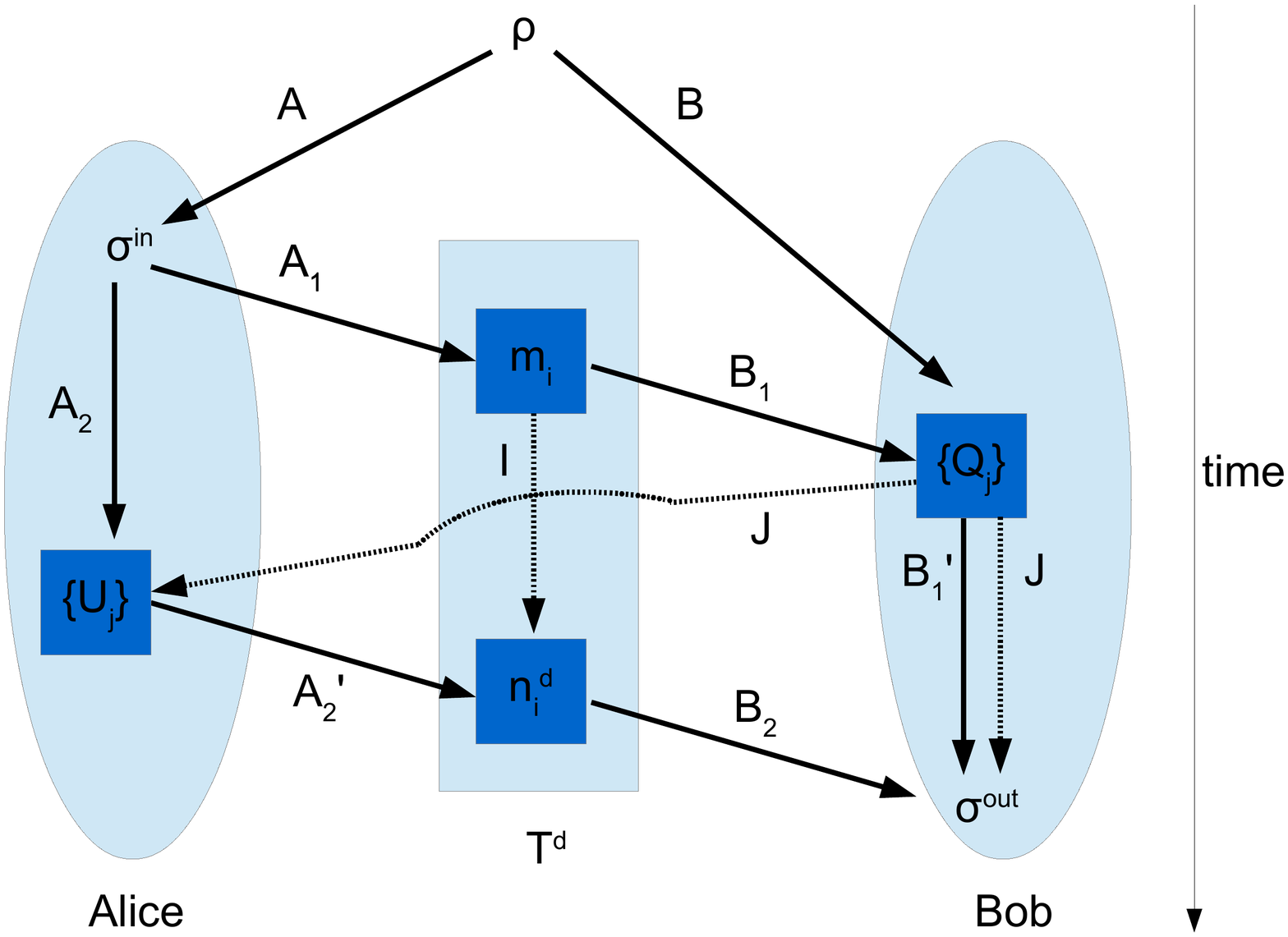}
		\caption{An even and odd round of our protocol, without (left) and with (right) entanglement assistance. The solid lines represent quantum systems, the dashed lines classical systems.}\label{timeflow}
\end{figure}

As mentioned above, the achievable rate of communication of this protocol greatly depends the ability to distinguish channels $m_0$ and $m_0$, which is where the entanglement assistance comes in. Without entanglement assistance, the maximum probability of distinguishing the two channels in one attempt is determined by the trace distance of the two channels
\be\label{Dist}
\left\|m_0-m_1\right\|_1=\max_{\sigma_A}\left\|m_0(\sigma_A)-m_1(\s_A)\right\|_1.
\ee
If, however, Alice and Bob are given a copy of an entangled state $\r$, the maximum success probability is determined by 
\be\label{AssDist}
\left\|m_0-m_1\right\|_\r=\left\|m_0\te\id(\rho_{AB})-m_1\te\id(\rho_{AB})\right\|_1.
\ee
We can now make use of the fact that any entangled state can provide an advantage in channel discrimination \cite{piani2009all}. Namely, a bipartite state $\rho_{AB}$ is entangled if and only if there exist two entanglement-breaking channels $m_{0,1}:A\to B_1$ such that
\be\label{EntAdv}
\left\|m_0-m_1\right\|_\r>\left\|m_0-m_1\right\|_1.
\ee
By choosing $d$ large enough compared to the output dimension of the $m_i$, we can achieve a locking effect, where even the smallest advantage a weakly entangled state can provide in channel discrimination can be amplified. This allows us to proof our main result:
\begin{theorem}\label{Theo:main}
For any entangled state $\r_{AB}$, there exists $d\in\NN$, such that 
\be
\cC_{\te,\r}^{\leftarrow}(T^d)>\cC_{\te}^{\leftarrow}(T^d).
\ee
\end{theorem}

Before proving Theorem \ref{Theo:main}, let us explicitly write down the feedback-assisted $n$-round Holevo capacity $\chi^{(n)}\left(T^d\right)$ of $T^d$ which we have defined above. For $n=1$, it reduces to the Holevo capacity of a mixture of $m_0$ and $m_1$,
\be\label{Holevo1}
\chi^{(1)}\left(T^d\right)=\chi\left(\frac{1}{2}(m_0+m_1)\right).
\ee
For two channel uses with feedback after the first channel use, it reads
\be\label{Holevo2}
\chi^{(2)}\left(T^d\right)=\max_{\left\{p_x,\omega^{(1,x)}_{A_1}\te\omega^{(2,xj_1)}_{A_2}\right\},\left\{\cE^{(1)}_{j_1}\right\}}I(X:B'_1B_2)_{\s^{(2)}},
\ee
for the cq-state  
\be\label{cq2}
\s^{(2)}_{XB'_1B_2}=\sum_{x}p_x\pro{x}_{X}\te \s^{(2,x)}_{B'_1B_2},
\ee
where
\be\label{cq2end}
\s^{(2,x)}_{B'_1B_2}=\frac{1}{2}\sum_{i_1=0}^1\sum_{j_1}p_{j_1|i_1x}\underbrace{\frac{1}{p_{j_1|i_1x}}\cE^{(1)}_{j_1}\circ m_{i_1}\left(\omega^{(1,x)}_{A_1}\right)}_{=:\tilde{\s}_{B'_1}^{(2,i_1j_1x)}}\te \underbrace{n^{d}_{i_1}\left(\omega^{(2,j_1x)}_{A_2}\right)}_{=:\hat{\s}_{B_2}^{(2,i_1j_1x)}},
\ee
where $p_{j_1|i_1x}=\tr\left[\cE^{(1)}_{j_1}\circ m_{i_1}\left(\omega^{(1,x)}_{A_1}\right)\right]$. Here $\{\cE^{(1)}_{j_1}\}$, with $\cE^{(1)}_{j_1}:B_1\to B'_1$, is a quantum instrument. The classical feedback consists of the outcome $j_1$ of the instrument. Depending on feedback $j_1$ and $x$, Alice inputs state $\omega^{(2,j_1x)}_{A_2}$. Note that whereas Bob can keep a copy of the classical feedback message $j_1$ this can be included into the output register $B'_1$ of the instrument. 

For $n>2$, we have
\be\label{eq:Holevon}
\chi^{(n)}\left(T^d\right)=\max I(X:B'_{n-1}B_n)_{\s^{(n)}},
\ee
where the maximisation is over product input ensembles $\left\{p_x,\omega^{(1,x)}_{A_1}\te\omega^{(2,j_1x)}_{A_2}\te\cdots\te\omega^{(n,j_{n-1}x)}_{A_n}\right\}$ as well as instruments $\left\{\cE^{(1)}_{j_1}\right\},...,\left\{\cE^{(n-1)}_{j_{n-1}}\right\}$ and the cq-state is given by
\be\label{eq:cqn}
\s^{(n)}_{XB'_{n-1}B_n}=\sum_{x}p_x\pro{x}_{X}\te\s^{(n,x)}_{B'_{n-1}B_n}.
\ee
For even $n>2$, $\s^{(n,x)}$ is defined recursively by
\be
\s^{(n,x)}_{B'_{n-1}B_n}=\frac{1}{2}\sum_{i_{n-1}=0}^1\sum_{j_{n-1}}p_{j_{n-1}|i_{n-1}x}\tilde{\s}_{B'_{n-1}}^{(n,i_{n-1}j_{n-1}x)}\te n_{i_{n-1}}^d\left(\omega^{(n,j_{n-1}x)}_{A_{n}}\right),
\ee
where
\begin{align}
&\tilde{\s}_{B'_{n-1}}^{(n,i_{n-1}j_{n-1}x)}:=\frac{1}{p_{j_{n-1}|i_{n-1}x}}\cE^{(n-1)}_{j_{n-1}}\left(\sum_{j_{n-2}}\cE^{(n-2)}_{j_{n-2}}\left(\s^{(n-2,x)}_{B'_{n-3}B_{n-2}}\right)\te m_{i_{n-1}}\left(\omega^{(n-1,j_{n-2}x)}_{A_{n-1}}\right)\right)\label{eq:29},
\end{align}
with
\be
p_{j_{n-1}|i_{n-1}x}=\tr\left[\cE^{(n-1)}_{j_{n-1}}\left(\sum_{j_{n-2}}\cE^{(n-2)}_{j_{n-2}}\left(\s^{(n-2,x)}_{B'_{n-3}B_{n-2}}\right)\te m_{i_{n-1}}\left(\omega^{(n-1,j_{n-2}x)}_{A_{n-1}}\right)\right)\right]
\ee
and  instruments
\begin{align}
&\cE^{(n-1)}_{j_{n-1}}:B'_{n-2}B_{n-1}\to B'_{n-1},\\
&\cE_{j_{n-2}}^{(n-2)}:B'_{n-3}B_{n-2}\to B'_{n-2},
\end{align}
that act collectively on all previous outcomes. For odd $n>2$, we define
\be
\s^{(n,x)}_{B'_{n-1}B_n}=\frac{1}{2}\sum_{i_n=0}^1\sum_{j_{n-1}} \cE^{(n-1)}_{j_{n-1}}\left(\sum_{j_{n-2}}\cE^{(n-2)}_{j_{n-2}}\left(\s^{(n-2,x)}_{B'_{n-3}B_{n-2}}\right)\te n^d_{i_n}\left(\omega^{(n-1,j_{n-2}x)}_{A_{n-1}}\right)\right)\te m_{i_{n}}\left(\omega^{(n,j_{n-1}x)}_{A_{n}}\right).
\ee


See also Figure \ref{setup}.

\begin{figure}
\centering
	\includegraphics[width=0.9\textwidth]{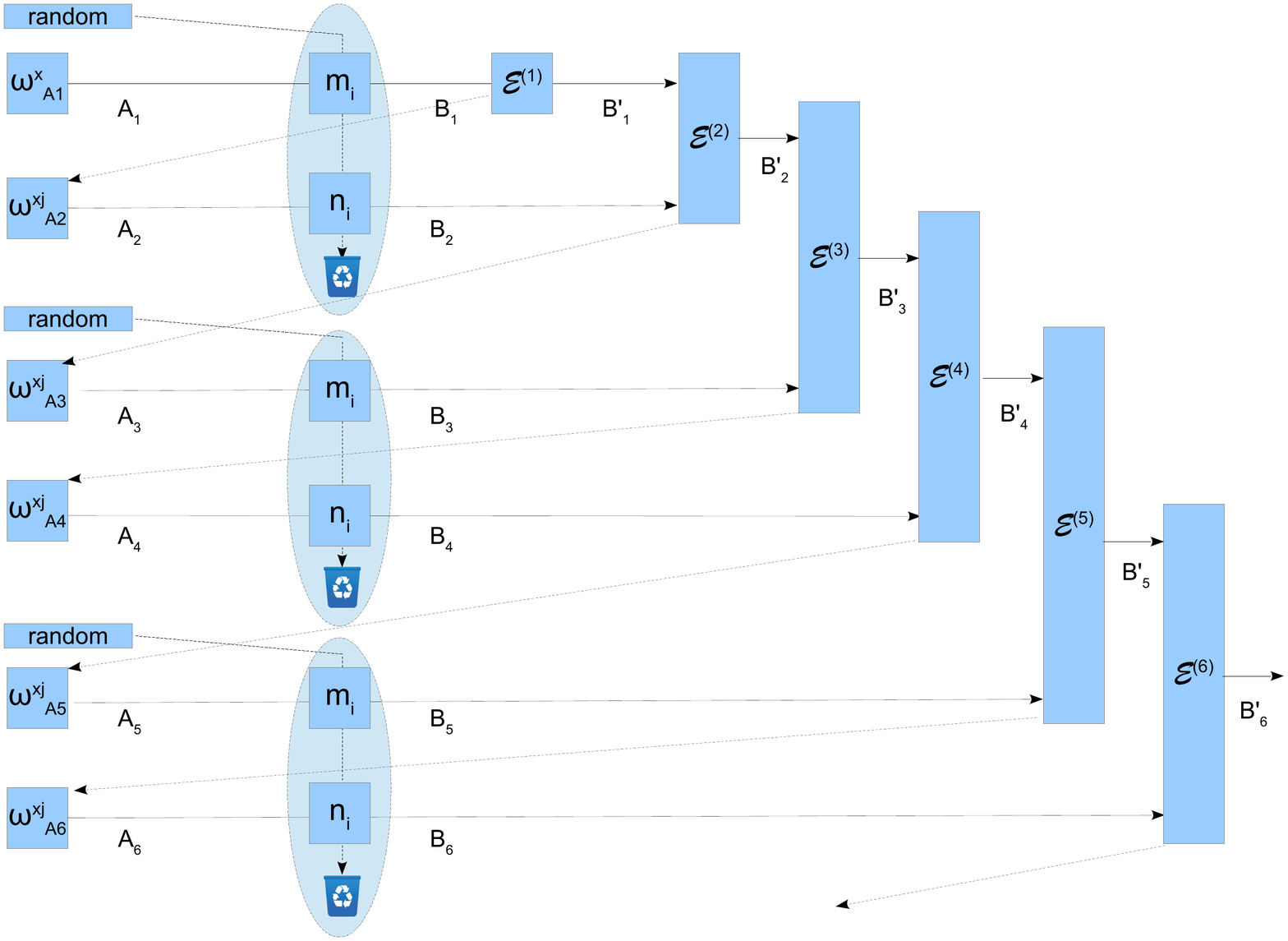}
	\caption{Our setting for feedback-assisted classical communication with product encoding via channel (\ref{eq:discChan}).}\label{setup}
\end{figure}

We are now ready to upper bound the classical capacity $\cC_{\te}^{\leftarrow}(T^d)$ of the unassisted scheme.
\begin{lemma}\label{L1}
The classical feedback product capacity of memory channel (\ref{eq:discChan}) is upper bounded as follows:
\be
\cC_{\te}^{\leftarrow}(T^d)\le\frac{1+\e}{4}\log{d}+\cO\left(\log \tilde{d}\right),
\ee
where $\e=\frac{1}{2}\left\|m_0-m_1\right\|_1$, $d=\text{dim}(B_2)$ and $\tilde{d}=\text{dim}(B_1).$
\end{lemma}

\begin{proof}
As the converse part of the HSW theorem \cite{1996quant.ph.11023H} can be extended to the case including feedback and memory, it holds
\be
\cC_{\te}^{\leftarrow}\left(T^d\right)\le\limsup_{n\to\infty}\frac{1}{n}\chi^{(n)}\left(T^d\right).
\ee
We will now prove by induction that for every even $n\ge2$ it holds
\be\label{InductionHypothesis2}
\chi^{(n)}\left(T^d\right)\le\frac{n}{2}\left(\frac{1+\e}{2}\log d +\cO\left(\log\tilde{d}\right)\right)+\cO\left(\log d\tilde{d}\right).
\ee
The base case $n=2$ follows trivially from the definition of $\chi^{(2)}$. Let us assume that (\ref{InductionHypothesis2}) holds for $n$. For the ensemble and operations maximising (\ref{eq:Holevon}), it holds
\begin{align}
\chi^{(n+2)}(T^d)&=I(X:B'_{n+1}B_{n+2})_{\s^{(n+2)}}\label{first:eq}\\
&=I(X:B'_{n+1})_{\s^{(n+2)}}+I(X:B_{n+2}|B'_{n+1})_{\s^{(n+2)}}\\
&\le\chi^{(n+1)}(T^d)+I(X:B_{n+2}|B'_{n+1})_{\s^{(n+2)}},\label{eq:56}
\end{align}
where we have made use of the data processing inequality. Performing the same procedure for $\chi^{(n+1)}$, we obtain
\begin{align}
\chi^{(n+1)}(T^d)&=I(X:B'_nB_{n+1})_{\s^{(n+1)}}\\
&=I(X:B'_n)_{\s^{(n+1)}}+I(X:B_{n+1}|B'_n)_{\s^{(n+1)}}\\
&\le\chi^{(n)}(T^d)+I(X:B_{n+1}|B'_n)_{\s^{(n+1)}}\\
&\le\chi^{(n)}(T^d)+\cO(\log \tilde{d}).\label{eq:60}
\end{align}
As for the second term in (\ref{eq:56}), let us note that the cq-state $\s^{(n+2)}_{XB'_{n+1}B_{n+2}}$, defined by (\ref{eq:cqn}) - (\ref{eq:29}), is separable w.r.t. the partition $XB'_{n+1}:B_{n+2}$.
This allows us to perform the following strategy \cite{bowen2005feedback} (in what follows we use indices $i:=i_{n+1}$, $j:=j_{n+1}$ and  $k:=j_{n}$).
\begin{align}
I(X:B_{n+2}|B'_{n+1})_{\s^{(n+2)}}&=S(B_{n+2}|B'_{n+1})_{\s^{(n+2)}}-S(B_{n+2}|XB'_{n+1})_{\s^{(n+2)}}\\
&\le \log d-\frac{1}{2}\sum_{i=0}^1\sum_{jx}p_xp_{j|ix}S(B_{n+2}|XB'_{n+1})_{\pro{x}\te\tilde{\s}^{(n+2,ijx)}\te n_{i}^d\left(\omega^{(n+2,jx)}\right)}\\
&=\log d-\frac{1}{2}\sum_{i=0}^1\sum_{jx}p_xp_{j|ix}S(B_{n+2})_{ n_{i}^d\left(\omega^{(n+2,jx)}\right)}\\
&=\log d-\sum_{jx}p_{jx}\sum_{i=0}^1p_{i|jx}S(B_{n+2})_{ n_{i}^d\left(\omega^{(n+2,jx)}\right)},\label{eq:46}
\end{align}

where, in the first inequality, we have made use of the concavity of the conditional entropy \cite{lieb2002proof}. Let us now consider a particular outcome $j$  and a particular message $x$. \emph{Conditioned} on $j$ and $x$, the probability that channel $m_{0,1}$ has been used can be expressed by 

\begin{align}
p_{0|jx}&=\frac{1}{2}+\frac{\e_{j|x}}{4p_{j|x}},\\
p_{1|jx}&=\frac{1}{2}-\frac{\e_{j|x}}{4p_{j|x}},
\end{align}
where 

\be
\e_{j|x}:=p_{j|0x}-p_{j|1x}=\tr\left[\cE^{(n+1)}_{j}\left(\sum_k\cE^{(n)}_{k}\left(\s^{(n,x)}_{B'_{n-1}B_{n}}\right)\te \left[m_0-m_1\right]\left(\omega^{(n+1,kx)}_{A_{n+1}}\right)\right)\right].
\ee

Going back to (\ref{eq:46}), let us recall 
that the channels $n_i^d$ are measurements in two MUBs. Note that, even if the probability of obtaining a particular outcome $j$ can depend on $i$, the state $\omega^{(n+2,jx)}_{A_{n+2}}$ itself, \emph{conditioned} on a particular outcome $j$, does not depend on $i$. Hence, we can apply an entropic uncertainty relation \cite{maassen1988generalized} in order to show that for every $j$ and $x$ it holds

\begin{align}
&\sum_{i=0}^1p_{i|jx}S\left(B_{n+2}\right)_{\hat{\s}^{(n+2,ijx)}}\\
&=\left(\frac{1}{2}+\frac{\e_{j|x}}{4p_{j|x}}\right) S\left(B_{n+2}\right)_{\hat{\s}^{(n+2,0jx)}}+\left(\frac{1}{2}-\frac{\e_{j|x}}{4p_{j|x}}\right)S\left(B_{n+2}\right)_{\hat{\s}^{(n+2,1jx)}}\\
&\ge \left(\frac{1}{2}-\frac{\left|\e_{j|x}\right|}{4p_{j|x}}\right) \log d.
\end{align}

Hence, we obtain
\be
I(X:B_{n+2}|B'_{n+1})_{\s^{(n+2)}}\le\left(\frac{1}{2}+\frac{1}{4}\sum_{jx}p_{jx}\frac{\left|\e_{j|x}\right|}{p_{j|x}}\right)\log d.
\ee
It holds further
\be
\sum_{jx}p_{jx}\frac{\left|\e_{j|x}\right|}{p_{j|x}}=\sum_{jx}p_{x}\left|\e_{j|x}\right|=\sum_{x}p_{x}\left[\sum_{j\in \cJ^{x}_0}\e_{j|x}-\sum_{j\in \cJ^{x}_1}\e_{j|x}\right],\label{eq54n}
\ee
where, given $x$, we have defined $\cJ^{x}_{0}=\left\{j:\e_{j|x}\ge0\right\}$ and $\cJ^{x}_{1}=\left\{j:\e_{j|x}<0\right\}$. Since $\sum_{j\in \cJ^{x}_0}\cE_{j}^{(n+1)}$ and $\sum_{j\in \cJ^{x}_1}\cE_{j}^{(n+1)}$ are sub-channels, they can be represented by Kraus operators as $\sum_{j\in \cJ^{x}_\ell}\cE_{j}^{(n+1)}(\cdot)=\sum_k F^{\ell x}_k(\cdot)F^{\ell x\dagger}_{k}$ where $F^{\ell x}:=\sum_kF^{\ell x\dagger}_kF^{\ell x}_{k}\le \1$ and $\ell=0,1$. Hence

\begin{align}
\sum_{j\in \cJ^{x}_0}\e_{j|x}&=\tr\left[F^{0x}\left(\sum_k\cE^{(n)}_{k}\left(\s^{(n,x)}_{B'_{n-1}B_{n}}\right)\te\left[m_0-m_1\right]\left(\omega^{(n+1,kx)}_{A_{n+1}}\right)\right)\right]\\
&\le\max_{\s_{AB}\text{ separable}} \max_{0\le\Lambda\le\1}\tr\left[\Lambda\left(m_0\te\id_B(\s_{AB})-m_1\te\id_B(\s_{AB})\right)\right]\\
&=\frac{1}{2}\left\|m_0-m_1\right\|_1=\e,
\end{align}

where we have used the fact that separable states cannot provide an advantage in channel discrimination \cite{piani2009all}. Up to a sign change the same holds true for the second term in (\ref{eq54n}), hence
\be
I(X:B_{n+2}|B'_{n+1})\le\frac{1+\e}{2}\log d.
\ee
Putting it all together, we obtain
\be\label{last:eq}
\chi^{(n+2)}(T^d)\le\chi^{(n)}(T^d)+\frac{1+\e}{2}\log d+\cO\left(\log\tilde{d}\right).
\ee
Application of the induction hypothesis completes the proof. 
\end{proof}
Let us now the consider the entanglement assisted case. 
\begin{lemma}\label{L2}
With $\d=\frac{1}{2}\left\|m_0-m_1\right\|_\r$, it holds
\be
\cC_{\te,\r}^{\leftarrow}(T^d)\ge\frac{1+\d}{2}\log{d}-1.
\ee
\end{lemma}
\begin{proof}
Let us assume that Alice and Bob share many copies of an entangled state $\r$. A possible strategy for Alice and Bob is the one we have already mentioned: In every odd round Alice enters her share of a copy of $\r$ into the channel. Bob tries to learn the value $i$ by means of a joint measurement of the output and his share of the copy of $\r$. His success probability is $\frac{1+\d}{2}$. He sends the result ($j$) to Alice, who then encodes her message in the corresponding MUB $\{\ket{v^{(j)}}\}$, which corresponds to transmission of the message via a classical symmetric channel with transmission probabilities $p_{y|x}=\frac{1+\d}{2}\d_{yx}+\frac{1-\d}{2d}$. This is a $d$-ary symmetric channel, whose capacity is well-known \cite{TC06}, and given by
\be
\cC(p_{y|x})=\log d-H\left(\underbrace{\frac{1-\d}{2d},...,\frac{1-\d}{2d}}_{d-1\text{ times}},\frac{1+\d}{2}+\frac{1-\d}{2d}\right).
\ee
As this is an achievable rate, it holds $\cC_{\te,\r}^{\leftarrow}(T^d)\ge \cC(p_{y|x})$. By the chain rule, it further holds 
\begin{align}
H\left(\frac{1-\d}{2d},...,\frac{1-\d}{2d},\frac{1+\d}{2}+\frac{1-\d}{2d}\right)&\le\left(1-\frac{1+\d}{2}-\frac{1-\d}{2d}\right)\log (d-1)+1\\
&\le\left(1-\frac{1+\d}{2}\right)\log d+1.
\end{align}
Hence
\be
\cC_{\te,\r}^{\leftarrow}(T^d)\ge\frac{1+\d}{2}\log d-1,
\ee
finishing the proof.
\end{proof}
Lemmas \ref{L1} and \ref{L2} allow us to prove the main result of this section.\newline

\begin{proof}{\bf of Theorem \ref{Theo:main}}
By Lemmas \ref{L1} and \ref{L2}, it holds
\be
\cC_{\te,\r}^{\leftarrow}(T^d)-\cC_{\te}^{\leftarrow}(T^d)\ge\frac{1}{2}(\d-\e)\log d-\cO(\log\tilde d),
\ee
where $\tilde{d}=\max\{\text{dim}(B_1),\text{dim}(B_1)\}$ is independent of $d$. By assumption, $\delta > \epsilon$, hence choosing $d$ large enough, we obtain
\be
\cC_{\te,\r}^{\leftarrow}(T^d)-\cC_{\te}^{\leftarrow}(T^d)>0,
\ee
finishing the proof.
\end{proof}

\section{Towards memoryless channels:\protect\\ An entropic entanglement witness}\label{sec:Shor}
Let us now return to the usual setting of many independent uses of a quantum channel and present our second approach.

The main result of this section is an entropic condition on a pair of a bipartite state $\r_{\tilde A\tilde B}$ and a channel $\mc{M}:\tilde A\to B$, which is sufficient for the existence of another channel $\mc{N}:\tilde AD\to B$, the Holevo capacity and, if $\mc{M}$ is entanglement breaking, classical capacity of which can be increased by $\r_{\tilde A\tilde B}$. As only entangled states can provide an advantage over the unassisted capacity, this result also provides a sufficient criterion for entanglement. Our reasoning begins with the observation that the minimum output entropy of a channel can be expressed as a conditional entropy:
\begin{lemma}\label{lemma:SSA}
For a channel $\mc{M}:\tilde A\to B$ it holds
\be\label{eq:SSA}
S_\text{min}(\mc{M})=\min_{\rho_{\tilde A\tilde B}\in\text{SEP}}S(B|\tilde B)_{\hat{\s}},
\ee
where $\hat{\s} = (M\otimes\id)(\rho)$ and $S_\text{min}(\mc{M}) = \min_\psi S(\mc{M}(\psi))$ is the minimum output entropy of $\mc{M}$.
\end{lemma}
\begin{proof}
By the strong subadditivity of the von Neumann entropy \cite{lieb2002proof}, the conditional entropy is concave, hence it holds
\begin{align}
\min_{\sigma\in\text{SEP}}S(B|\tilde B)_{\hat{\s}}&\ge\min_{\{p_i,\ket{\phi_i},\ket{\varphi_i}\}}\sum_ip_iS(B|\tilde B)_{\mc{M}(\pro{\phi_i})\te\pro{\varphi_i}}\\
&=\min_{\{p_i,\ket{\phi_i}\}}\sum_ip_iS\left(\mc{M}(\pro{\phi_i})\right)\\
&\ge\min_{\ket{\phi}} S(\mc{M}(\ket{\phi})).
\end{align}
The converse follows from the subadditivity of the von Neumann entropy
\begin{align}
\min_{\sigma\in\text{SEP}}S(B|\tilde B)_{\hat{\s}}&\le\min_{\sigma\in\text{SEP}}S(B)_{\hat{\s}}\\
&=\min_{\{p_i,\ket{\phi_i}\}}S\left(\sum_ip_i\mc{M}(\pro{\phi_i})\right)\\
&\le\min_{\ket{\phi}} S(\mc{M}(\ket{\phi})),
\end{align}
finishing the proof.
\end{proof}
Lemma \ref{lemma:SSA} provides us with an entropic entanglement witness \cite{bauml2015witnessing,rastegin2015uncertainty}: If $S(B|\tilde B)_{\hat{\s}}<S_\text{min}(\mc{M})$, then $\r_{\tilde A \tilde B}$ has to be entangled. We will now show that that this witness has an operational interpretation. 
\begin{theorem}\label{lemma:CondHolevo}
For every state $\rho_{\tilde A \tilde B}$ and channel $\mc{M}:\tilde A\to B$, such that 
\begin{equation}\label{eq:Cond}
  S(B|\tilde B)_\omega < S_\text{min}(\mc{M}),
\end{equation}
there exists a channel $\mc{N}:A\to B$, such that $\chi_{\rho} (\mc{N})>\chi (\mc{N})$.
\end{theorem}
\begin{proof}
Assume we have a channel  $\mc{M}:\tilde A\to B$, fulfilling eq. (\ref{eq:Cond}). Following \cite{shor2004equivalence}, it is possible to construct a channel $\mc{N}:A\to B$, such that 
\be
\chi(\mc{N}) = \log |B| - S_\text{min}(\mc{M}).
\ee
$\mc{N}$ consists of $\mc{M}$ and an additional input $D$ with $|D|=|B|^2$. Define $A=\tilde A\te D$. System $D$ is dephased and acts as a control of generalised Pauli matrices \cite{PR04} acting on the output system $B$ of $\mc{M}$. Hence
\be
\mc{N}(\rho_{\tilde A}\te\pro{jk}_{D})=U^{jk}\mc{M}(\rho_{\tilde A})U^{jk\dagger},
\ee
for $0\le j,k\le|B|$. See also Figure \ref{fig:ShorN}. It is now easy to see that the ensemble 
\be
\left\{\frac{1}{|B|^2},\pro{v^\text{min}_{\tilde A}}\te\pro{jk}_{D}\right\},
\ee
where $\ket{v^\text{min}}$ minimises the output entropy of $\mc{M}$, maximises the Holevo capacity. Namely, the complete output state gets maximally mixed, maximising the first term in eq. (\ref{def:Holevo2}), whereas application of a single unitary does not change the entropy of $\mc{M}(\pro{v^\text{min}})$, minimising the second term.\newline
\begin{figure}
	\centering
	\includegraphics[width=0.4\textwidth]{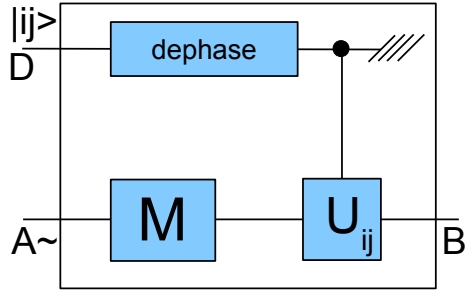}
	
	\caption{Channel $\mc{N}:\tilde AD\to B$ as in \cite{shor2004equivalence}. Input $D$ is dephased, then acts as control on generalised Pauli unitaries acting on the output of $\mc{M}$.}\label{fig:ShorN}
\end{figure}
Let us now lower bound the $\rho$-assisted Holevo capacity of $\mc{N}$. A feasible, however not necessarily optimal, encoding strategy for Alice is the following: She leaves her share of $\rho_{\tilde A\tilde B}$ untouched and, as in the unassisted case, inserts $\pro{jk}$ into the $D$ entry of the channel with equal probability $p_{ij}=\frac{1}{|B|^2}$, i.e.
\be
\cE_{jk}^{\tilde A\to D\tilde A}\te\id_{\tilde B}(\rho_{\tilde A\tilde B})=\pro{jk}_{D}\te\rho_{\tilde A\tilde B}.
\ee
The action of $\mc{N}$ results in $\tilde{\rho}_{B\tilde B}=\1_B\te\rho_{\tilde B}$ and $\tilde{\rho}^{jk}_{B\tilde B}=\left(U^{jk}_B\te\1_{\tilde B}\right)\mc{M}(\hat{\s})\left(U^{jk}_B\te\1_{\tilde B}\right)^\dagger$. Hence, since unitaries do not change the entropy,
\be
\chi_\r(\mc{N})\ge \log |B|+S(\tilde B)_\rho-S(CB\tilde B)_\omega= \log |B|-S(B|\tilde B)_\omega.
\ee
Hence, if eq. (\ref{eq:Cond}) is fulfilled for $\mc{M}$, it holds $\cC_\r(\mc{N})\ge\chi_\r(\mc{N})>\chi(\mc{N})$ finishing the proof. 
\end{proof}

If, in addition, $\mc{M}$ is entanglement-breaking, so will be $\mc{N}$. Hence by \cite{shor2002additivity} it holds $\chi(\mc{N})=\cC(\mc{N})$ and we can formulate the following
\begin{corollary}\label{lemma:Cond}
For every state $\rho_{\tilde A\tilde B}$ and entanglement-breaking channel $\mc{M}:\tilde A\to B$, such that 
condition (\ref{eq:Cond}) is fulfilled, there exists a channel $\mc{N}:A\to B$, such that $\cC_{\rho} (\mc{N})>\cC (\mc{N})$.
\end{corollary}

The minimum output entropy of a channel is, in general, difficult to compute. In fact it has been shown to be NP-hard \cite{beigi2007complexity}. There exist, however, channels for which it is easy to compute. One example is the depolarising channel \cite{king2003capacity}
\be\label{eq:depol}
\mc{M}_D(\m)=t\m+\frac{(1-t)}{d}\1,
\ee
where $-\frac{1}{d^2-1}\le t\le1$. It is easy to see that for pure input states, the output spectrum, hence the output entropy, only depends on the parameter $t$, not on the input state. Namely $\lambda(\mc{M}_D(\pro{\phi}))=\left(t+\frac{1-t}{d},\frac{1-t}{d},...,\frac{1-t}{d}\right)$ for all $\ket{\phi}$, hence $S_{\min}(\mc{M}_D)=H\left(t+\frac{1-t}{d},\frac{1-t}{d},...,\frac{1-t}{d}\right)$. However, depolarising channels cannot be used to witness bound entanglement, as the following result shows.

\begin{lemma}\label{lemma:specCond}
For channels $\mc{M}:\tilde A\to B$ such that $\lambda(\omega_{B\tilde B})$ only depends on the marginal spectra $\lambda(\rho_{\tilde A\tilde B})$, $\lambda(\rho_{\tilde B})$ or $\lambda(\rho_{\tilde A})$ of $\rho_{\tilde A\tilde B}$, no bound entangled $\rho_{\tilde A\tilde B}$ can fulfil eq. (\ref{eq:Cond}).
\end{lemma}
\begin{proof}
Let $\rho_{\tilde A\tilde B}$ be bound entangled. By \cite{hiroshima2003majorization}, this implies that $\lambda(\rho_{\tilde B})\succ\lambda(\rho_{\tilde A\tilde B})$ and $\lambda(\rho_{\tilde A})\succ\lambda(\rho_{\tilde A\tilde B})$. Hence, by Theorem 2 of \cite{nielsen2001separable}, there exists a separable state $\tilde{\sigma}_{\tilde A\tilde B}$ such that $\lambda(\tilde{\sigma}_{\tilde A\tilde B})=\lambda(\rho_{\tilde A\tilde B})$, $\lambda(\tilde{\sigma}_{\tilde B})=\lambda(\rho_{\tilde B})$ and $\lambda(\tilde{\sigma}_{\tilde A})=\lambda(\rho_{\tilde A})$. Let us define $\tilde{\s}_{B\tilde B}=(M\otimes\id)\tilde{\sigma}_{\tilde A\tilde B}$. By assumption, it holds $\lambda(\omega_{B\tilde B})=\lambda(\tilde{\s}_{B\tilde B})$, hence 
\be
S(B|\tilde B)_\omega=S(B|\tilde B)_{\tilde{\s}}\ge\min_{\sigma_{\tilde A\tilde B}\in\text{SEP}}S(B|\tilde B)_{M\te\id(\sigma)}=S_\text{min}(\mc{M}),
\ee
where the last equality is due to Lemma \ref{lemma:SSA}.
\end{proof}

It is possible to overcome the limitations arising from Lemma \ref{lemma:specCond}, by adding a transposition to the depolarising channel. Namely, we consider the \emph{transpose depolarising channel} \cite{fannes2004additivity}
\be\label{eq:depolT}
\mc{M}_D^T(\m)=t\m^T+\frac{(1-t)}{d}\1,
\ee
for $-\frac{2}{d^2-2}\le t\le\frac{1}{d+1}$. For $-\frac{1}{d^2-1}\le t\le\frac{1}{d+1}$, (\ref{eq:depolT}) is entanglement-breaking. As transposition of a pure input state does not change the spectrum, the transpose depolarising channel has the same minimum output entropy as the depolarising channel. In contrast to the depolarising channel, the global spectrum of the output states does not only depend on the marginal spectra of the inputs, hence the transpose depolarising channel can be used to activate bound entanglement in the sense of Corollary \ref{lemma:Cond}.

Another example where $S_{\min}$ is known is the two-Pauli channel \cite{bennett1997entanglement} acting on a two qubit system 
\be\label{eq:2Pauli}
M_{XZ}(\m)=t\m+\frac{1-t}{2}\s_x\m\s_x+\frac{1-t}{2}\s_z\m\s_z.
\ee
Using the Bloch sphere representation of $\rho$, it is straightforward to show that for $t=1/3$ $\lambda(M_{XY}(\rho))=(\frac{2}{3},\frac{1}{3})$ for all input states, hence $S_\text{min}(M_{XZ})=H(\frac{2}{3},\frac{1}{3})$. 
Other examples of channels with known minimum output entropy are given in \cite{wolf2005classical}. 

\subsection*{Numerical examples}
As proof of principle examples, we have numerically checked the condition (\ref{eq:Cond}) for the two-qubit Werner state (\ref{eq:Werner}) and a number of channels, for which the minimum output entropy is known. The two-qubit Werner state is not one-way distillable for $q\ge1/4$. In this parameter region, namely for $1/4\le q\le0.345$, we could observe a fulfillment of condition (\ref{eq:Cond}) of Corollary \ref{lemma:Cond} for the depolarising channel (\ref{eq:depol}) with channel parameter $t=-1/3$, the transpose depolarising channel (\ref{eq:depolT}) with $t=1/3$ as well as the two-Pauli channel (\ref{eq:2Pauli}) with $t=1/3$, showing that there exists a state that is not one-way distillable but can still provide an advantage in classical communication via the channels arising by using those channels in the Shor construction. In figure \ref{fig:plots} we have plotted the difference $\Delta S:=S(B|\tilde B)_\omega-S_{\min}$ versus the Werner state parameter $q$, for the depolarizing channel with $t=-1/3$. Where $\Delta S<0$, the condition (\ref{eq:Cond}) is fulfilled. 

\begin{figure}
	\centering
	\includegraphics[trim=1cm 5cm 2cm 5cm,clip=true,width=0.5\textwidth]{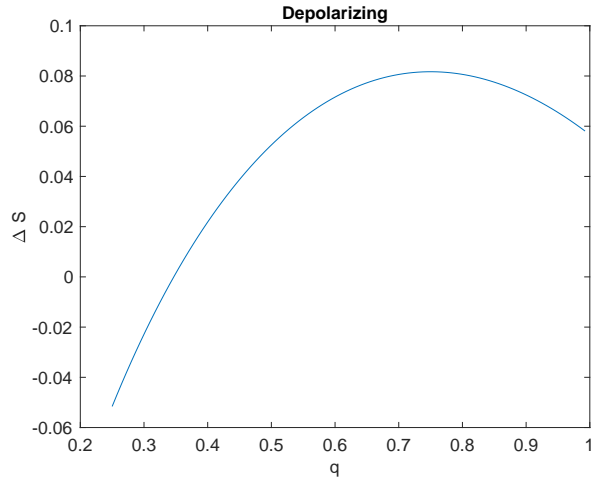}
	\caption{$\Delta S$ plotted versus state parameter $q$ for the two-qubit Werner state (\ref{eq:Werner}) and the depolarising channel  (\ref{eq:depol}) with $t=-1/3$. Where $\Delta S <0$, Theorem \ref{lemma:CondHolevo} can be applied.  The transpose depolarising channel (\ref{eq:depolT}) and the two-Pauli channel (\ref{eq:2Pauli}) with $t=1/3$ provide the same curves.}\label{fig:plots}
	\end{figure}

\section{Summary and Outlook}
We have studied the role noisy entanglement with nonpositive coherent information can play in classical communication via a noisy quantum channel. We have provided a memory channel where any entangled state can provide an advantage if we allow for feedback from Bob and Alice and restrict to product encodings. The advantage is obtained by employing the fact that any entangled state can provide an advantage in channel discrimination. In future work we plan to investigate different communication scenarios where the gap between the assisted and unassisted case can be increased. In a second approach we have constructed a channel, the Holevo capacity of which can be increased whenever an entropic condition is met. Namely, if we are given a channel $\mc{M}$ and an entangled state $\r$, we let random unitaries act on the output of $\mc{M}$, such that the combined channel $\mc{N}$, with output system $B$, has Holevo capacity $\chi(\mc{N})=\log |B|-S_\text{min}(\mc{M})$, where $S_{min}$ denotes the minimum output entropy. The $\r$-assisted Holevo capacity, on the other hand, is given by $\chi_\r(\mc{N})=\log |B|-S(B|\tilde B)_{\hat{\s}}$, where $\omega_{B\tilde B}=\id_{\tilde B}\te \mc{M}^{\tilde A\to B}(\r_{\tilde A\tilde B})$. Comparison of the two provides us with an entropic criterion that is sufficient for an increase in the Holevo capacity and, if $\mc{M}$ is entanglement-breaking, the classical capacity by using $\r$ as an additional resource. As separable states cannot provide an advantage in classical communication this criterion can also serve as an entanglement witness. Using this criterion and the (transpose) depolarising channel as well as the two Pauli channel as choices for $\mc{M}$, we have shown that Werner states that are not one-way distillable, hence cannot be used in dense coding with a noiseless channel, can provide an advantage in classical communication. We plan to further exploit this criterion and hope to find more examples of states and channels showing nontrivial results. One way to achieve this could be to minimise $\Delta S$ over (entanglement-breaking) channels for given states, which might be challenging as the objective already contains a nontrivial optimisation problem.              

\section*{Acknowledgements}
The authors are indebted to many colleagues for interesting discussions on entanglement-assistance over the years. Among others, these include Marcus Huber, Bill Munro, Go Kato, Koji Azuma, Kae Nemoto, Fabian Furrer, Yanbao Zhang, Elton Yechao Zhu and Paul Skrzypczyk. The authors acknowledge support by the European Commission (STREP ``RAQUEL''), the ERC (Advanced Grant ``IRQUAT''), the Spanish MINECO (grants FIS2008-01236, FIS2013-40627-P
and FIS2016-86681-P), with the support of FEDER funds, and by the Generalitat de Catalunya, CIRIT project 2014-SGR-966, the NFR Project ES564777, by the NSFC (grant nos. 11375165, 11875244) as well as by the Netherlands Organization for Scientific Research (NWO/OCW), as part of the Quantum Software Consortium programme (project number 024.003.037 / 3368) and an NWO Vidi grant.

\bibliographystyle{unsrt}
\bibliography{EntAsCap.bib}
\end{document}